\documentclass[a4paper,UKenglish,cleveref, autoref, thm-restate]{lipics-v2021}

\usepackage{tikz-cd}

\pdfoutput=1 
\hideLIPIcs

\title{The Complexity of Promise Constraint Satisfaction Problem Seen from the Other Side}

\author{Kristina Asimi}{Department of Algebra, Faculty of Mathematics and Physics, Charles University, Czechia \and \url{https://sites.google.com/view/kristina-asimi/home}}{asimptota94@gmail.com}{}{}

\author{Libor Barto}{Department of Algebra, Faculty of Mathematics and Physics, Charles University, Czechia \and \url{https://www2.karlin.mff.cuni.cz/~barto/} }{libor.barto@gmail.com}{https://orcid.org/0000-0002-8481-6458}{}

\author{Victor Dalmau}{Department of Information and Communication Technologies, Universitat Pompeu Fabra, Spain \and \url{https://www.upf.edu/web/victor-dalmau}}{victor.dalmau@upf.edu}{}{}

\authorrunning{K. Asimi, L. Barto, and V. Dalmau}

\Copyright{Kristina Asimi, Libor Barto, and Victor Dalmau}



\category{} 

\relatedversion{}

\funding{Kristina Asimi and Libor Barto have received funding from the European Research Council (ERC) under the European Unions Horizon 2020 research and innovation programme (Grant Agreement No.~771005, CoCoSym). Libor Barto was also funded by ERC, POCOCOP, 101071674. Victor Dalmau was funded by the MiCin (PID2019-109137GB-
C/AEI/10.13039/501100011033).}

\nolinenumbers

\usepackage{tikz-cd}

\newtheorem{question}[theorem]{Question}

\usepackage{makecell}

\newcommand{\CSP}{\mathrm{CSP}}
\newcommand{\PCSP}{\mathrm{PCSP}}
\newcommand{\rel}[1]{\mathbb{#1}}

\newcommand{\tuple}[1]{\mathbf{#1}}

\newcommand{\ar}{\mathrm{ar}}

\newcommand{\yes}{\mathtt{Yes}}
\newcommand{\no}{\mathtt{No}}

\newcommand{\HOM}{\mathrm{Hom}}

\newcommand{\clique}{\mathtt{Clique}}

\newcommand{\PHOM}{\mathrm{PHom}}

\newcommand{\gapc}{\mathtt{Gap}$-$\mathtt{Clique}}

\newcommand{\leftcomp}[1]{#1^{\leftarrow}}
\newcommand{\rightcomp}[1]{#1^{\rightarrow}}
\newcommand{\size}[1]{|\!|#1|\!|}

\newcommand{\uuu}[1]{U\langle #1 \rangle}
\newcommand{\vvv}[1]{V\langle #1 \rangle}

\begin{document}

\maketitle

\begin{abstract}
We introduce the framework of the left-hand side restricted promise constraint satisfaction problem, which includes problems like approximating clique number of a graph. We study the parameterized complexity of problems in this class and provide some initial results. The main technical contribution is a sufficient condition for W[1]-hardness which, in particular, covers left-hand side restricted bounded arity CSPs.
\end{abstract}

\section{Introduction}

The Constraint Satisfaction Problem (CSP) provides a
common framework for expressing a wide range of both theoretical and real-life
combinatorial problems~\cite{RBW06}. One solves an instance of the CSP by assigning
values to its variables so that its constraints are satisfied.

A convenient way to formalize the CSP is via homomorphisms between relational structures.
Informally (see Section \ref{sec:prelim} for formal definitions), a \emph{relational structure} $\rel A$ consists of a collection of relations on a common universe $A$. Two structures are \emph{similar} if they have the same number of relations and the corresponding relations are of the same arity. A \emph{homomorphism} between similar structures is a relation-preserving map between their universes.
The CSP can then be defined as the problem to decide whether there is a homomorphisms $\rel A \to \rel B$ between a given pair $(\rel A,\rel B)$ of similar structures.

\subsection{Right-hand side restrictions}

Many computational problems, including various versions of logical satisfiability, graph coloring, and systems of equations can be phrased as CSPs with a fixed right-hand side structure $\rel A$, also called the \emph{template} of the CSP (see~\cite{BKW17}). 

An example of such a CSP is \emph{$k$}-$\mathtt{coloring}$, the problem of deciding whether the vertices of a given graph can be colored by $k$ different colors so that no adjacent vertices are assigned the same color. Formulated as a homomorphism problem, \emph{$k$}-$\mathtt{coloring}$ is the problem of deciding whether a given graph admits a homomorphism to the $k$-clique $\rel K_k$ (complete graph on $k$ vertices). It follows that the \emph{$k$}-$\mathtt{coloring}$ problem is equivalent to the CSP with template $\rel K_k$. As is well known, this problem is solvable in polynomial time for $k=2$ and it is NP-complete for $k>2$~\cite{Kar72}. 

The computational complexity of CSPs over finite templates (i.e., templates whose universe is a finite set) is now completely classified by a celebrated dichotomy theorem independently obtained by Bulatov~\cite{Bul17} and Zhuk~\cite{Zhu17,Zhu20}: every such problem is either tractable (that is, solvable in polynomial-time) or NP-complete.

A recent direction of research (see \cite{BBKO21,KO22}) is to study a generalization of fixed-template CSPs called the Promise CSP (PCSP). A \emph{template} for the PCSP is a pair $(\rel A,\rel B)$ of similar structures such that $\rel A$ has a homomorphism to $\rel B$, and the PCSP over $(\rel A,\rel  B)$ is distinguishing between the case that a given finite structure $\rel X$ admits a homomorphism to $\rel A$ and the case that $\rel X$ does not even admit a homomorphism to $\rel B$ (the promise is that one of the cases takes place). The PCSP framework generalizes that of CSP (take $\rel A = \rel B$) and includes important approximation problems. For example, if $\rel A = \rel K_{k}$  and $\rel B = \rel K_l$, $k \leq l$, then $\PCSP$ over $(\rel A,\rel B)$ is a version of the \emph{approximate graph coloring problem}, namely, the problem to distinguish graphs that are $k$-colorable from those that are not even $l$-colorable. The classification of the complexity of this problem is an open problem after more than 40 years of research.

\subsection{Left-hand side restrictions}

By restricting the left-hand side in the CSP instead of the right-hand side, we get problems of different flavor: we restrict the structure of the constraints rather than their language. For instance, if we fix the left-hand side structure to $\rel K_k$ we essentially get the problem  of deciding whether a given graph contains a $k$-clique. This problem is tractable, as well as any  CSP with a fixed finite left-hand side structure, since we can simply check all possible mappings to the right-hand side structure.

However, restricting the left-hand side to a \emph{class} of structures rather than a single structure yields interesting problems. For example, if we take the class of all finite cliques, we essentially get the problem to decide, for a given graph and integer $k$, whether the graph contains a $k$-clique. This NP-complete problem \cite{Kar72} plays a major role in the parameterized complexity theory \cite{downey2012parameterized}, e.g., as a starting point for hardness results. 

For classes of structures $\mathcal{C}$ of bounded arity, a dichotomy theorem was provided in Grohe's paper \emph{The complexity of homomorphism and constraint satisfaction problems seen from the other side} \cite{Grohe2007}: under some complexity theoretic assumption, the CSP with the left-hand side restricted to $\mathcal C$ is solvable in polynomial time if and only if $\mathcal{C}$ has bounded tree width modulo homomorphic equivalence.

\subsection{Contribution}

We generalize the left-hand side restricted CSPs to left-hand side restricted PCSPs in a way similar to the right-hand restrictions discussed above. 
Namely, a \emph{template} is a class of pairs of structures $(\rel A, \rel B)$ such that there is a homomorphism from $\rel A$ to $ \rel B$. The problem is to decide, for a given $(\rel A,\rel B)\in \mathcal{C}$ and a relational structure $\rel X$ similar to $\rel A$ (and $\rel B$), whether there is a homomorphism from $\rel B$ to $\rel X$ or not even from $\rel A$ to $\rel X$. Examples include versions of the \emph{approximate clique problem}: given $k \geq l$ (coming from some fixed set) and a graph, decide whether the graph has a $k$-clique or not even an $l$-clique. 

Our main technical contribution is a sufficient condition for hardness that generalizes the hardness part of the above mentioned Grohe's result \cite{Grohe2007} to the promise setting (and also somewhat simplifies it). 

The next two sections are devoted to the left-hand side restricted CSP. The presentation largely follows \cite{Grohe2007}. The last section is devoted to the left-hand side restricted PCSP.

\section{Preliminaries} \label{sec:prelim}

\subsection{Relational structures and homomorphisms}

A \emph{signature} is a finite collection of relation symbols each with an associated \emph{arity}, denoted $\ar(R)$ for a relation symbol $R$. The \emph{arity} of a signature is the maximum of the arities of all relations symbols it contains. A \emph{relational structure} $\rel A$ in the signature $\sigma$, or a $\sigma$-\emph{structure}, consists of a finite set $A$, called the \emph{universe} of $\rel A$, and a relation $R^{\rel A} \subseteq A^{\ar(R)}$
for each symbol $R$ in $\sigma$, called the \emph{interpretation} of $R$ in $\rel A$. Two structures are called \emph{similar} if they are in the same signature.  
We say that a class $\mathcal{C}$ of structures is of \emph{bounded arity} if there is an $r$ such that arity of the signature of every structure in $\mathcal{C}$ is at most $r$.

A structure over a signature containing a single binary relation symbol is called a \emph{directed graph}, or \emph{digraph}. If this relation is symmetric and loop free (i.e., it contains no pairs of the form $(a,a)$), we call the structure \emph{undirected graph}. If the relation of a graph is the disequality relation on the universe, we call the graph a \emph{complete graph} or \emph{clique}.

A $\sigma$-structure
$\rel A$ is a \emph{substructure} of a $\sigma$-structure
$\rel B$, denoted by $\rel A \subseteq \rel B$, if $A \subseteq B$ and $R^\rel A \subseteq R^\rel B$ for all $R \in \sigma$.
A structure $\rel A$ is a \emph{proper} substructure of $\rel B$, denoted by $\rel A \subset \rel B$, if $\rel A \subseteq \rel B$ and $\rel A \neq \rel B$.

We define the size of a $\sigma$-structure $\rel A$ to be
\[\size{\rel A} = |\sigma | + |A| + \Sigma _{R\in\sigma} |R^A| ar(R).\]
$\size{\rel A}$ is roughly the size of a reasonable encoding of $\rel A$. 
When taking structures $\rel A$ as inputs for algorithms, we measure the running time of the algorithm in terms of $\size{\rel A}$.

Given two similar structures $\rel A$ and $\rel B$, a  function $f$ from $A$ to $B$ is called a \emph{homomorphism} from $\rel A$ to $\rel B$ if $f(\tuple{a}) \in R^{\rel B}$ for any $\tuple{a} \in R^{\rel A}$, where $f(\tuple{a})$ is computed component-wise.
If there exists a homomorphism from $\rel{A}$ to $\rel{B}$, we write $\rel{A}\rightarrow\rel{B}$, and if there is none, we write $\rel{A}\not\to\rel{B}$. The composition of homomorphisms is a homomorphism.

Two structures $\rel A$ and $\rel B$ are \emph{homomorphically equivalent} if $\rel{A}\rightarrow\rel{B}$ and $\rel{B}\rightarrow\rel{A}$.

A relational structure $\rel A$ is a \emph{core} if there is no homomorphism from $\rel A$ to a proper
substructure of $\rel A$. A core \emph{of} a structure $\rel A$ is a substructure $\rel A'$ of $\rel A$ such that $\rel A\to\rel A'$
and $\rel A'$
is a core. Obviously, every core of
a structure is homomorphically equivalent to the structure. It can be shown that all cores of a structure $\rel A$ are isomorphic. So, we often speak of \emph{the} core of $\rel A$.

\subsection{Homomorphism problem}

General \emph{homomorphism problem} (or \emph{constraint satisfaction problem}) asks whether there is a homomorphism from one structure to another. We are interested in restrictions of this problem.
For two classes $\mathcal{C}$ and $\mathcal{D}$ of structures, $\HOM(\mathcal{C},\mathcal{D})$ is the following problem.

\medskip

Input: Similar structures $\rel A \in \mathcal{C}$, $\rel B \in \mathcal{D}$;

Output: $\yes$ if $\rel A\to \rel B$; $\no$ if $\rel A\not\to \rel B$.\\

If $\mathcal{C}$ is the class of all finite structures, we write $\HOM (-,\mathcal{D})$ instead of $\HOM(\mathcal{C},\mathcal{D})$. The problem $\HOM(-,\{\rel A\})$ is also known as $\CSP (\rel A)$.

Similarly, if $\mathcal{D}$ is the class of all finite structures, we write $\HOM (\mathcal{C},-)$ instead of $\HOM(\mathcal{C},\mathcal{D})$, and we call such a problem the \emph{left-hand side restricted CSP}. 
 If $\mathcal{C}$ is finite, then $\HOM (\mathcal{C},-)$ is solvable in polynomial time, so usually we are interested in the case when $\mathcal{C}$ is an infinite collection of structures.

If $\mathcal{C}$ is the class of all cliques, the problem $\HOM (\mathcal{C},-)$ is called the \emph{$\clique$ problem}.
In \cite{Kar72} Karp proved that $\clique$ is NP-complete.

\subsection{Graph minors and tree width}

We will denote the vertex set of a graph $\rel G$ by $G$ and its relation (or the set of edges) by $E^\rel G$. Since we are considering undirected graphs, we will view its edges as sets (unordered pairs) $e=\{v,w\}$, and we will use notations like $v\in e$ or $\{v,w\}\in E^\rel G$.

A graph $\rel H$ is a \emph{minor} of a graph $\rel G$ if $\rel H$ is isomorphic to a graph that can be obtained from a subgraph of $\rel G$ by contracting edges.
For example, in Figure \ref{fig:minors}, it is easy to see that $\rel H$ is a minor of $\rel G$.

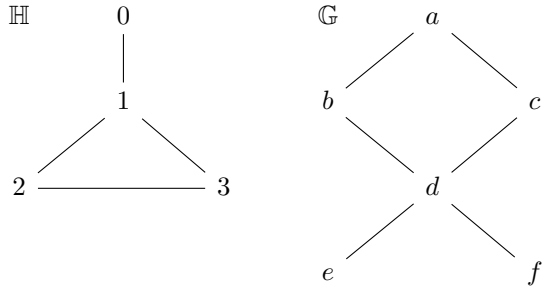
\begin{figure}[ht]
\begin{center}
\begin{tikzcd}
\rel H                     & 0 \arrow[d, no head]                      &   & \rel G                     & a \arrow[ld, no head] \arrow[rd, no head]                     &   \\
                      & 1 \arrow[ld, no head] \arrow[rd, no head] &   & b \arrow[rd, no head] &                                                               & c \\
2 \arrow[rr, no head] &                                           & 3 &                       & d \arrow[ru, no head] \arrow[ld, no head] \arrow[rd, no head] &   \\
                      &                                           &   & e                     &                                                               & f
\end{tikzcd}
\end{center}
\caption{\label{fig:minors}An example of a minor of a graph.}

\end{figure}

    A \emph{minor map} from $\rel H$ to $\rel G$ is
a mapping $\mu : H \to 2^G$ with the following properties.
\begin{itemize}
\item For all $v \in H$, the set $\mu(v)$ is nonempty and connected in $\rel G$.
\item For all $v,w \in H$, with $v \neq w$, the sets $\mu(v)$ and $\mu(w)$ are disjoint.
\item For all edges $\{v,w\} \in E^\rel H$, there are $v'\in\mu(v)$ and $w'\in\mu(w)$ such that $\{v',w'\} \in E^\rel G$.
\end{itemize}

For any two graphs $\rel H$ and $\rel G$, there is a minor map from $\rel H$ to $\rel G$ if and only if $\rel H$ is a minor of $\rel G$.
Moreover, if $\rel H$ is a minor of a connected graph $\rel G$, then we can always find a minor map from $\rel H$ onto $\rel G$, where by \emph{onto} we mean
\[
\bigcup_{v\in H}\mu(v)=G.\]
Going back to Figure \ref{fig:minors}, an example of a minor map from $\rel H$ to $\rel G$ is $\mu(0)=\{e\}$, $\mu(1)=\{d\}$, $\mu(2)=\{b\}$, $\mu(3)=\{a,c\}.$

    \emph{Trees} are connected acyclic graphs.
A \emph{tree-decomposition} of a graph $\rel G$ is a pair
$(\rel T , \beta)$, where $\rel T$ is a tree and $\beta : T \to 2^G$
such that the following conditions
are satisfied:
\begin{itemize}
\item For every $v\in G$ the set $\{t \in T | v \in \beta (t)\}$ is non-empty and connected in $\rel T$.
\item For every $e \in E^\rel G$ there is a $t \in T$ such that $e \subseteq \beta (t)$.
\end{itemize}
For example, for every graph $\rel G$ there is a tree-decomposition where the tree is one vertex mapping to the whole set $G$.

The \emph{width} of a tree-decomposition $(\rel T,\beta)$ is $\max\{|\beta (t)||t\in T\}-1$, and the
\emph{tree width} of a graph $\rel G$, denoted by $tw(\rel G)$, is the minimum $w$ such that $\rel G$ has a
tree-decomposition of width $w$.

For $k,l\geq 1$, the \emph{$(k \times l)$-grid} is the graph with vertex set $[k] \times [l]$ and an edge between $(i,j)$ and $(i',j')$ if and only if $|i-i'|+|j-j'| = 1$.
It can be shown that the
$(k \times k)$-grid has tree width $k$. Figure \ref{fig:k-grid} shows a tree-decomposition of width 3 for $(3\times 3$)-grid.
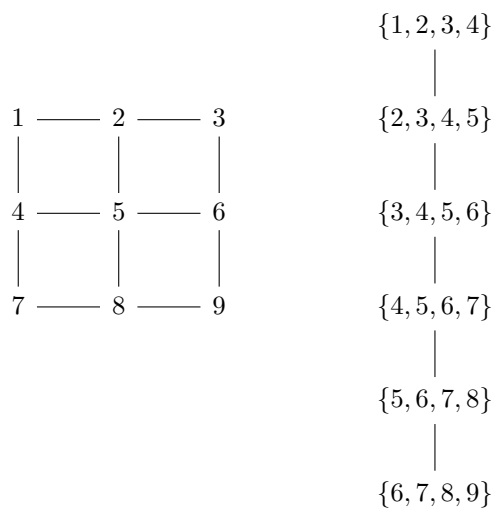
\begin{figure}[ht]
\begin{center}
\begin{tikzcd}
                                        &                                         &                      &  & \{1,2,3,4\} \arrow[d, no head] \\
1 \arrow[r, no head] \arrow[d, no head] & 2 \arrow[r, no head] \arrow[d, no head] & 3 \arrow[d, no head] &  & \{2,3,4,5\} \arrow[d, no head] \\
4 \arrow[r, no head] \arrow[d, no head] & 5 \arrow[r, no head] \arrow[d, no head] & 6 \arrow[d, no head] &  & \{3,4,5,6\} \arrow[d, no head] \\
7 \arrow[r, no head]                    & 8 \arrow[r, no head]                    & 9                    &  & \{4,5,6,7\} \arrow[d, no head] \\
                                        &                                         &                      &  & \{5,6,7,8\} \arrow[d, no head] \\
                                        &                                         &                      &  & \{6,7,8,9\}                   
\end{tikzcd}
\end{center}
\caption{\label{fig:k-grid}A tree-decomposition of a $(3\times 3$)-grid.}

\end{figure}
In \cite{ROBERTSON198692} Robertson and Seymour proved the following
“converse”, which is known as the Excluded Grid Theorem.

\begin{theorem}\label{thm:excludedgrid}
For every $k$ there exists a
$w(k)$ such that the $(k\times k)$-grid is a minor of every graph of tree width at least $w(k)$.
\end{theorem}

We need to transfer some of the graph theoretic notions to arbitrary relational structures.
The \emph{Gaifman graph} (also known as \emph{primal graph}) of a $\sigma$-structure $\rel A$ is the graph $\mathcal{G}(\rel A)$ with vertex set $A$ and an edge between vertices $a$ and $b$ if $a \neq b$
and there is a relation symbol $R\in\sigma$, say, of arity $r$, and a tuple $(a_1,a_2,\dots, a_r) \in R^\rel A$
such that $a, b \in \{a_1,a_2,\dots, a_r\}$. 
 We can now transfer the notions of graph minor theory to
relational structures. In particular, a subset $B \subseteq A$ is \emph{connected} in a relational structure $\rel A$ if it
is connected in $\mathcal{G}(\rel A)$. A \emph{minor map} from a relational structure $\rel A$ to a relational structure $\rel B$ is a mapping $\mu: A \to 2^B$ that is a minor map from $\mathcal{G}(\rel A)$ to $\mathcal{G}(\rel B)$. A \emph{tree decomposition} of a relational structure $\rel A$ can simply be defined to be a tree-decomposition of $\mathcal{G}(\rel A)$. 

A class $\mathcal{C}$ of structures has \emph{bounded tree width} if there exists $w$ such that $tw(\rel A) \leq w$ for all $\rel A \in \mathcal{C}$. A class $\mathcal{C}$ of structures has \emph{bounded tree width modulo
homomorphic equivalence} if there is $w$ such that every $\rel A \in \mathcal{C}$ is homomorphically equivalent to a structure of tree width at most $w$.

\subsection{Parameterized complexity theory}

Parameterized complexity theory studies the complexity of decision problems with respect to both the size of the input and an additional parameter. 

Formally, a \emph{parameterization} of a decision problem $P\subseteq \Sigma ^*$, where $\Sigma$ is an alphabet, is a polynomial time computable mapping $\kappa : \Sigma ^*\to\mathbb{N}$, and a \emph{parameterized problem} over $\Sigma$ is a pair $(P, \kappa)$ consisting of a problem $P\subseteq \Sigma ^*$ and a parameterization $\kappa$ of $P$. For example, the \emph{parameterized clique problem p-}$\clique$ is the problem $(P,\kappa)$, where $P$ is the set of all pairs $(\rel G,k)$ (suitably encoded over some finite alphabet) such that $\rel G$ contains a $k$-clique and the parameter $\kappa$ is defined by $\kappa (\rel G,k):=k$. We present parameterized problems in
the following form.

\smallskip

\emph{p-}$\clique$:

Input: Graph $\rel G$, $k \in \mathbb{N}$;

Parameter: $k$;

Output: $\yes$ if $\rel G$ has a clique of size $k$; $\no$ otherwise.\\

\noindent
We will parameterize the homomorphism problem in the following way.

\smallskip

\emph{p-}$\HOM (\mathcal{C},\mathcal{D})$:

Input: Similar structures $\rel A\in\mathcal{C}$, $\rel B\in\mathcal{D}$;

Parameter: $\size{\rel A}$;

Output: $\yes$ if $\rel A\to\rel B$; $\no$ otherwise.\\

As before, if $\mathcal{D}$ is the set of all finite relational structures, we write \emph{p-}$\HOM(\mathcal{C},-)$.

A parameterized problem $(P, \kappa)$ over $\Sigma$ is \emph{fixed-parameter tractable} if there is a
computable function $f:\mathbb{N}\to\mathbb{N}$ and an algorithm that decides if a given instance $x\in\Sigma ^*$ belongs to the problem $P$ in time
\[f(\kappa (x)) |x|^{O(1)}.\]
The class of all fixed-parameter tractable parameterized problems is denoted by FPT.

While in the traditional theory of computational complexity we use polynomial-time reductions, in the theory of parameterized complexity we require an analogous notion of reduction that preserves fixed-parameter tractability, called an \emph{fpt-reduction}.

An \emph{fpt-reduction} from a parameterized problem $(P, \kappa)$ over $\Sigma$ to a parameterized
problem $(P', \kappa ')$ over $\Sigma '$ is a mapping $R:\Sigma ^*\to (\Sigma ')^*$ such that:
\begin{itemize}
\item For all $x\in\Sigma ^*$ we have $x\in P$ if and only if $R(x)\in P'$.

\item There is a computable function $f:\mathbb{N}\to\mathbb{N}$ and an algorithm that, given $x\in \Sigma ^*$, computes $R(x)$ in time $f(\kappa (x)) |x|^{O(1)}$.

\item There is a computable function $g:\mathbb{N}\to\mathbb{N}$ such that for all instances $x\in \Sigma ^*$ we have $\kappa '
(R(x)) \leq g(\kappa(x))$.
\end{itemize}.

A specific example of an fpt-reduction is the following reduction of \emph{p-}$\clique$ to \emph{p-}$\HOM (\mathcal{C},-)$, where $\mathcal{C}$ is a class of structures that contains all cliques: an instance $(\rel G,k)$ of \emph{p-}$\clique$ is mapped to the instance $(\rel K _k,\rel G)$ of \emph{p-}$\HOM (\mathcal{C},-)$.

\emph{Hardness} and \emph{completeness} of parameterized problems for a parameterized complexity class are defined in the usual way: a parameterized problem is said to be hard for a parameterized complexity class if every problem in that class can be reduced to it using an fpt-reduction. A parameterized problem is said to be complete for a parameterized complexity class if it is hard for that class and it belongs to that class.

An analogue of NP in parameterized complexity theory is the class W[1]. It is a widely believed standard assumption that FPT $\neq$ W[1]. For the technical definitions of this and other parameterized complexity classes see \cite{downey2012parameterized} or \cite{FG2006}. 

It was shown that \emph{p-}$\clique$ is W[1]-complete \cite{DOWNEY1995109}.
The significance of this problem in the area of parameterized complexity is fundamental, as it has been used as a starting point for numerous reductions. We will see one of them in the next section.

\section{Complexity of the left-hand side restricted CSP}

In this section we present a dichotomy theorem~\cite{Grohe2007} for decidable classes of structures of bounded arity: for each such a class $\mathcal{C}$, \emph{p}-$\HOM(\mathcal{C},-)$ is either in FPT (even solvable in polynomial time) or W[1]-complete.

Marx's deep paper~\cite{Marx13} investigates the complexity for the case of unbounded arity. He also argues
that investigating the fixed-parameter
tractability of left-hand side restrictions is at least as interesting as investigating polynomial-time solvability. One reason is that FPT seems to be a more robust class in this context, e.g., these problems are unlikely to exhibit a standard, non-parameterized complexity dichotomy~\cite{Grohe2007}. 

The tractability part of the bounded arity dichotomy has been proved by Dalmau, Kolaitis, and Vardi in \cite{Dalmau2002}. We state a slightly stronger version as stated in \cite{Grohe2007}.

\begin{theorem}\label{thm:tractabilitydalmau}
Let $\mathcal{C}$ be a class of structures of bounded
tree width modulo homomorphic equivalence. Then $\HOM(\mathcal{C},-)$ is in polynomial time.
\end{theorem}

In \cite{Grohe2007} Grohe proved the following hardness theorem.

\begin{theorem}\label{thm:hardGrohe}
Let $\mathcal{C}$ be a recursively enumerable class of structures of bounded
arity that does not have bounded tree width modulo homomorphic equivalence.
Then \emph{p-}$\HOM(\mathcal{C},-)$ is \emph{W[1]}-hard under fpt-reductions.
\end{theorem}

We will sketch the proof, but we need some preparation first.

Let $k\geq 2$ and $K=\binom{k}{2}$ and let $\rel A$ be a connected $\sigma$-structure with a $(k\times K)$-grid as a minor in its Gaifman graph. Since $A$ is connected, there is a minor map from the $(k\times K)$-grid onto $\rel A$. We fix such a map $\mu$. We also fix some bijection $\rho$ between $[K]$ and the set of all unordered pairs of elements of $[k]$. For $p \in [K]$ we will write $i\in p$ instead of $i\in\rho(p)$. It will be convenient to switch between viewing the columns of the $(k\times K)$-grid as being indexed by elements of $[K]$ and unordered pairs of distinct elements of $[k]$.

For a graph $\rel G=(G,E^\rel G)$ we will define a $\sigma$-structure $\rel M = \rel M(\rel A,\mu,\rel G)$ such
that there exists a homomorphism from $\rel A$ to $\rel M$ if and only if $\rel G$ contains a $k$-clique. We define the universe of $\rel M$ to be
\begin{align*}
M = \{(v, e,i, p, a) | & v \in G, e \in E^\rel G, \\ & i \in [k], p \in [K] \text{ such that } (v \in e \iff i \in p), \\ &
a \in \mu (i, p)\}.
\end{align*}
We define the \emph{projection} $\Pi :M\to A$ by letting
\[
\Pi (v, e,i, p, a) = a \]
for all $(v, e,i, p, a) \in M$. 

We shall define the relations of $\rel M$ in such a way that $\Pi$
is a homomorphism from $\rel M$ to $\rel A$. For every $R\in\sigma$ , say, of arity $r$, and for all
tuples $\mathbf{a} = (a_1,a_2,\dots, a_r) \in R^\rel A$ we add to $R^\rel M$ all tuples $\mathbf{m} = (m_1,m_2,\dots, m_r)\in\Pi ^{-1}(\mathbf{a})$ satisfying the following two constraints for all $m, m'\in\{m_1,m_2,\dots , m_r\}$:
\begin{itemize}
    \item If $m=(v,e,i,p,a)$ and $m'=(v',e',i,p',a')$, then $v=v'$.
    \item If $m=(v,e,i,p,a)$ and $m'=(v',e',i',p,a')$, then $e=e'$.
\end{itemize}

The following lemmas will guarantee correctness of the reduction in the proof of Theorem~\ref{thm:hardGrohe}.

\begin{lemma}\label{lem:cliquetohom}
    If $\rel G$ contains a $k$-clique, then $\rel A\to\rel M$.
\end{lemma}

\begin{lemma}\label{lem:homtoclique}
    Suppose that $\rel A$ is a core. If $\rel A\to\rel M$, then $\rel G$ contains a $k$-clique.
\end{lemma}

\begin{proof}[Sketch of the proof of Theorem \ref{thm:hardGrohe}]
We will give an fpt-reduction from \emph{p-}$\clique$ to \emph{p-}$\HOM (\mathcal{C},-)$. Let $\rel G$ be a graph and let $k\geq 1$. Let $K=\binom{k}{2}$. By the Theorem \ref{thm:excludedgrid}, there exists a structure $\rel A\in\mathcal{C}$ such that the Gaifman graph of the core of $\rel A$ has the $(k\times K)$-grid as a minor.

We find such an $\rel A$,
compute the core $\rel A'$ of $\rel A$ and a minor map $\mu$ from the $(k \times K)$-grid to $\rel A'$. 
We let $\rel A''$ be the connected component of $\rel A'$ that contains the image of $\mu$. $\rel A''$ is also a core. We can assume, without loss of generality, that $\mu$ is a minor map
from the $(k \times K)$-grid onto $\rel A''$. We let $\rel M '=\rel M (\rel A '',\mu ,\rel G)$.  By Lemma \ref{lem:cliquetohom} and Lemma \ref{lem:homtoclique}, $\rel A ''\to\rel M '$ if and only if $\rel G$ contains a $k$-clique. Let $\rel M$ be a disjoint union of $\rel M'$ and $\rel A' \backslash \rel A ''$. Since $\rel A '$ is a core, every homomorphism from $\rel A'$ to $\rel M$ maps $\rel A ''$ to $\rel M '$. So, $\rel A '\to\rel M$ if and only if $\rel G$ contains a $k$-clique. Since $\rel A '$ is the core of $\rel A$, it means that $\rel A\to\rel M$ if and only if $\rel G$ has a $k$-clique.
\end{proof}

The main result of \cite{Grohe2007}
is the following theorem, which combines the previous two (Theorem \ref{thm:tractabilitydalmau} and Theorem \ref{thm:hardGrohe}). Note that for a decidable class $\mathcal{C}$, \emph{p-}$\HOM (\mathcal{C},-)$ is in W[1], so W[1]-hardness becomes W[1]-completeness.

\begin{theorem}
    Assume that \emph{FPT} $\neq$ \emph{W[1]}. Then, for every recursively enumerable class $\mathcal{C}$ of structures of bounded arity the following statements are 
equivalent:
\begin{itemize}
    \item $\HOM(\mathcal{C},-)$ is in polynomial time.
\item \emph{p-}$\HOM (\mathcal{C},-)$ is fixed-parameter tractable.
\item $\mathcal{C}$ has bounded tree width modulo homomorphic equivalence.
\end{itemize}
If either statement is false, then \emph{p-}$\HOM (\mathcal{C},-)$ is \emph{W[1]}-hard.
\end{theorem}

\section{Left-hand side restricted PCSP}

General \emph{promise homomorphism problem} is the following parameterized problem:

\medskip

Input: Similar structures $\rel A, \rel B, \rel X$ such that $\rel A \to \rel B$;

Parameter: $\size{A} + \size{B}$;

Output: $\yes$ if $\rel B\to \rel X$; $\no$ if $\rel A\not\to \rel X$.\\

It is a \emph{promise problem} in that the sets of $\yes$ and $\no$ instances do not cover the set of all inputs and there are no requirements on algorithms if the input does not fall into $\yes$ or $\no$ (or, put differently, the computer is promised that the input is $\yes$ or $\no$ and the task is to decide which of the options take place)
\footnote{
In fact, the general homomorphism problem and its restrictions should be also regarded as promise problems -- we are promised that the input is in the expected form; see~\cite{Grohe2007} for a discussion about this issue.  
}. On the other hand, we do require that the set of $\yes$ instances is disjoint from the set of $\no$ instances. For the general promise homomorphism problem, this property is guaranteed by the requirement $\rel A \to \rel B$. 

The definition of an fpt-reduction naturally extends to promise problems: Instead of the first condition, we require that $\yes$-instances are mapped to $\yes$-instances (\emph{completeness}) and that $\no$-instances are mapped to $\no$-instances (\emph{soundness}). Since the definition of $\no$-instances often involves negation (e.g. in the general homomorphism problem), soundness is often shown by proving the contrapositive: if the image is not a $\no$-instance, then neither is the original instance.

We define the left-hand side restricted PCSP as the general promise homomorphism problem restricted to a class of pair of structures $(\rel A,\rel B)$. We will be only concerned with recursively enumerable classes of bounded arity and so we include this requirement.

\begin{definition}
A collection of pairs of similar structures $(\rel A,\rel B)$ such that $\rel A \to \rel B$ is called a \emph{template} if it is recursively enumerable and of bounded arity.

For a template $\Gamma$, the \emph{left-hand side restricted PCSP over $\Gamma$}, denoted $\PHOM(\Gamma,-)$, is the following problem.

\smallskip

Input: Similar structures $\rel A,\rel B,\rel X$ where $(\rel A,\rel B)\in\Gamma$;

Parameter: $\size{A} + \size{B}$;

Output: $\yes$ if $\rel B\to\rel X$; $\no$ if $\rel A\not\to\rel X$.
\end{definition}

It is clear that the left-hand side restricted PCSP is a generalization of the (bounded arity, recursively enumerable) left-hand side restricted CSP: \emph{p-}$\HOM (\mathcal{C},-)$ is equivalent to $\PHOM(\Gamma,-)$ for $\Gamma = \{(\rel A,\rel A) \mid \rel A \in \mathcal{C}\}$. We do not distinguish between $\mathcal{C}$ and the template $\Gamma$ in this situation and call this template a \emph{CSP template}.  

Important examples of left-hand side restricted PCSP are approximation versions of the $\clique$ problems. For a given computable function $f: \mathbb{N} \to \mathbb{N}$ with $f(n)<n$, $n \in \mathbb{N}$, the $f$-$\gapc$ problem is: given a graph $\rel G$ and $k \in \mathbb{N}$ decide whether $\rel G$ has a $k$-clique or not even an $f(k)$-clique. Clearly, $f$-$\gapc$ is equivalent to the $\PHOM(\Gamma,-)$ for $\Gamma = \{(\rel K_{f(k)},\rel K_k) \mid k \in \mathbb{N}\}$. We discuss this class of problems in the last subsection.

\subsection{Homomorphic relaxations}

A simple, but important reduction for the right-hand side restricted PCSP is by means of homomorphic relaxation (see \cite{BBKO21}). A natural left-hand side version of this concept is as follows.

\begin{definition} 
Let $\Gamma$ and $\Delta$ be templates. We say that $\Gamma$ is a \emph{(left-hand side) homomorphic relaxation} of $\Delta$ if for every $(\rel A,\rel B) \in \Gamma$ there exists $(\rel C,\rel D) \in \Delta$ such that all four structures are similar and $\rel A \to \rel C$ and $\rel D \to \rel B$.
\end{definition}

\begin{proposition}
Let $\Gamma$ and $\Delta$ be templates. If $\Gamma$ is a homomorphic relaxation of $\Delta$, then $\PHOM(\Gamma)$ is fpt-reducible to $\PHOM(\Delta)$.
\end{proposition}

\begin{proof}
We map an instance $\rel A, \rel B, \rel X$ of $\PHOM(\Gamma,-)$ to the instance $\rel C, \rel D, \rel X$ of $\PHOM(\Delta,-)$, where $\rel C, \rel D$ are chosen so that $\rel A \to \rel C$ and $\rel D \to \rel B$.

If $\rel A, \rel B, \rel X$ is a $\yes$ instance of $\PHOM(\Gamma)$, then $\rel B \to \rel X$. Therefore, $\rel D \to \rel X$ as $\rel D \to \rel B$ and homomorphisms compose, so $\rel C, \rel D, \rel X$ is a $\yes$ instance of $\PHOM(\Delta)$. This shows completeness of the reduction.
Similarly, if $\rel C, \rel D, \rel X$ is not a $\no$ instance of $\PHOM(\Delta)$, then $\rel C \to \rel X$, so $\rel A \to \rel X$ (as $\rel A \to \rel C$), thus $\rel A, \rel B, \rel X$ is not a $\no$ instance of $\PHOM(\Gamma)$, showing soundness. The remaining requirements on fpt-reduction are clear.
\end{proof}

By Theorem~\ref{thm:tractabilitydalmau}, $\HOM(\mathcal{C},-)$ is in polynomial time whenever $\mathcal{C}$ is of bounded tree width. We immediately obtain the following corollary.

\begin{corollary} \label{cor:relax}
    If a template $\Gamma$ is a homomorphic relaxation of a CSP template of bounded tree width, then $\PHOM(\Gamma,-)$ is fixed-parameter tractable.
\end{corollary}

Note that a CSP template $\Gamma$ is a homomorphic relaxation of a CSP template of bounded tree width if and only if $\Gamma$ has bounded tree width modulo homomorphic equivalence. The tractability condition from Corollary~\ref{cor:relax} is by Theorem~\ref{thm:hardGrohe} the only source of fixed-point tractability (assuming FPT $\neq$ W[1]) for CSP templates. We do not have a good reason to believe that this result generalizes to general templates, but neither do we have a counter-example, and the following question thus arises.  

\begin{question}
 Let $\Gamma$ be a template which is not a homomorphic relaxation of a CSP template of bounded tree width. Is then $\PHOM(\Gamma,-)$ necessarily W[1]-hard?   
\end{question}

\subsection{Sufficient condition for hardness}

In this subsection we improve the sufficient condition for W[1]-hardness from Theorem~\ref{thm:hardGrohe} and give some corollaries.

The construction in the proof is largely inspired by Grohe's construction presented in the previous section, but there are several differences. First, we remove some unnecessary components, namely $i,p$ in the definition of $M$. On the other hand, we allow constant number of components of type $v,e$ to add more flexibility.
Third, we formulate the hardness criterion so that it can be directly applied e.g. to all the left-hand side restricted CSPs, not just connected cores. 

The fourth difference is in that we use $(k \times k)$-grids instead of $(k \times \binom{k}{2})$-grids, which makes the construction somewhat more natural. In fact, instead of minor maps from grid, we use a more general concept of grid-like mappings that we now introduce.
We use the following convention. If $\rho$ is a mapping from $C$ to a product $D \times E$, we denote by $\leftcomp{\rho}$ and $\rightcomp{\rho}$ its left and right components, respectively. That is, $\leftcomp{\rho}: C \to D$ and $\rightcomp{\rho}: C \to E$ are such that $\rho(c) = (\leftcomp{\rho}(c),\rightcomp{\rho}(c))$.

\begin{definition}
    Let $\rel C$ be a structure. A mapping $\rho: C \to [k] \times [k]$ is called a \emph{grid-like mapping} from $\rel C$ onto $[k] \times [k]$ if it is surjective and, for each $i \in [k]$, both $(\leftcomp{\rho})^{-1}(\{i\})$ and $(\rightcomp{\rho})^{-1}(\{i\})$ are connected subsets of the Gaifman graph of $\rel C$. 
\end{definition}

Note that a minor map $\mu$ from a $(k \times k)$-grid onto a structure $\rel C$ gives rise to a grid-like mapping $\rho$ from $\rel C$ onto $[k] \times [k]$ by defining $\rho(c)$ as the unique pair $(i,j)$ such that $c \in \mu(i,j)$.

\begin{theorem} \label{thm:left_hard}
    Let $\Gamma$ be a template. Suppose there exists $L \in \mathbb{N}$ such that for every $k \in \mathbb{N}$ the following condition is satisfied:

    \begin{description}
    \item[(*)] 
    There exists $(\rel A,\rel B) \in \Gamma$ and mappings $\rho_1, \rho_2, \dots, \rho_L: B \to [k] \times [k]$ such that for each homomorphism $g: \rel A \to \rel B$ there exists a structure $\rel C$, a homomorphism $h: \rel C \to \rel A$, and $l \in [L]$ such that $\rho_l g h$ is a grid-like mapping  from $\rel C$ onto $[k] \times [k]$. 
    \end{description}

    Then $\PHOM(\Gamma,-)$ is $W[1]$-hard.
\end{theorem}

\begin{proof}
    Let $L$ be as in the statement.
    We will give an fpt-reduction from \emph{p-}$\clique$ to $\PHOM (\Gamma,-)$.
    Let $\rel G$ be a graph and let $k \geq 1$. Let $\sigma$-structures $\rel A$ and $\rel B$ and mappings $\rho_i$ be as in (*). We map the instance $(\rel G,k)$ of \emph{p-}$\clique$ to the instance $(\rel A, \rel B, \rel X)$, where $\rel X$ is the $\sigma$-structure constructed as follows. 

    We define the universe of $\rel X$ to be
    \begin{align*}
    X & = \{(b, (u_i,v_i)_{i=1}^L) \mid \ b \in B, \mbox{ for all $i \in [L]$ } (u_i,v_i) \in G \times G \mbox{ such that } \\
      &  \leftcomp{\rho}_i(b) = \rightcomp{\rho_i}(b) \Rightarrow u_i=v_i \mbox{ and } 
       \leftcomp{\rho}_i(b) \neq \rightcomp{\rho_i}(b) \Rightarrow \{u_i,v_i\} \in E^{\rel G}
      \} \\ & \subseteq B \times (G \times G)^{L}.
    \end{align*}
   We define the projection $\Pi: X \to B$ by
   $$
   \Pi(b, (u_i,v_i)_{i=1}^L) = b
   $$
   for all $(b, (u_i,v_i)_{i=1}^L) \in X$. We define the relations of $\rel X$ in such a way that $\Pi$ is a homomorphism from $\rel X$ to $\rel B$. For every symbol $R \in \sigma$, say, of arity $r$, and for all tuples $\tuple{b} = (b_1, b_2, \dots, b_r) \in R^{\rel B}$ we add to $R^{\rel X}$ all tuples $\tuple{x} = (x_1, x_2, \dots, x_r) \in \Pi^{-1}(\tuple{b})$ satisfying the following two constraints for all $(b, (u_i,v_i)_{i=1}^L), (b', (u'_i,v'_i)_{i=1}^L) \in \{x_1, x_2, \dots, x_r\}$ and all $i \in [L]$:
   \begin{description}
       \item[(Cl)] If $\leftcomp{\rho_i}(b) = \leftcomp{\rho_{i}}(b')$, then $u_i = u_i'$ .
       \item[(Cr)] If $\rightcomp{\rho_i}(b) = \rightcomp{\rho_{i}}(b')$, then $v_i = v_i'$ .
   \end{description}
   Note that $\Pi$ is indeed a homomorphism from $\rel X$ to $\rel B$ (even without imposing the constraints (Cl) and (Cr)).

\medskip

   The completeness of the reduction is guaranteed by the following claim.
   \begin{claim}
   If $\rel G$ contains a $k$-clique, then $\rel B \to \rel X$.
   \end{claim}
   \claimproof
      Let $v\langle 1 \rangle, v\langle 2\rangle, \dots, v\langle k\rangle \in G$
      be vertices of a $k$-clique in $\rel G$. We define $h: B \to X$ by
      $$
      h(b) = (b, (v\langle\leftcomp{\rho}_i(b)\rangle,v\langle\rightcomp{\rho}_i(b)\rangle )_{i=1}^L).
      $$
      We need to verify that $h(b)$ is indeed in $X$, i.e., that $\leftcomp{\rho}_i(b)=\rightcomp{\rho_i}(b)$ implies $v\langle \leftcomp{\rho}_i(b)\rangle =v\langle \rightcomp{\rho}_i(b)\rangle $, and that $\leftcomp{\rho}_i(b) \neq \rightcomp{\rho_i}(b)$ implies $\{v\langle \leftcomp{\rho}_i(b)\rangle,v\langle \rightcomp{\rho}_i(b)\rangle \} \in E^{\rel G}$. Both implications are immediate.

      In order to check that $h$ is a homomorphism, take arbitary $R \in \sigma$ and $\tuple{b}=(b_1, b_2, \dots, b_r) \in R^{\rel B}$. Clearly, $h(\tuple{b}) \in \Pi^{-1}(\tuple{b})$. Moreover, for any $j,j' \in [r]$, if $\leftcomp{\rho}_i(b_j) = \leftcomp{\rho}_i(b_{j'})$, then trivially $v\langle\leftcomp{\rho}_i(b_j)\rangle=v\langle\leftcomp{\rho}_i(b_{j'})\rangle$, so (Cl) is satisfied. Similarly, (Cr) is satisfied as well, therefore $h(\tuple{b}) \in R^{\rel X}$. This finishes the proof of the claim.

\medskip

   In order to prove the soundness of the reduction, assume that there exists a homomorphism from $\rel A$ to $\rel X$, say $\alpha: \rel A \to \rel X$. We need to find a $k$-clique in $\rel G$. Let $g = \Pi\alpha$ and let $\rel C$, $h$, and $l \in [L]$ be as in (*), i.e., $\xi$ defined by
   $$\xi = \rho_l\Pi\alpha h$$ 
   is a grid-like mapping from $\rel C$ to $[k] \times [k]$.

   For each $c \in C$, the element $\alpha h(c)$ is of the form $\alpha h(c) = (b, (u_i,v_i)_{i=1}^L) \in X$. We set $u\langle c\rangle = u_l$ and $v\langle c\rangle = v_l$. The definition of $X$ implies the following claim.

   \begin{claim}
   If $\leftcomp{\xi}(c) = \rightcomp{\xi}(c)$, then $u\langle c\rangle =v\langle c\rangle $.    
   If $\leftcomp{\xi}(c) \neq \rightcomp{\xi}(c)$, then $\{u\langle c\rangle ,v\langle c\rangle\} \in E^{\rel G}$.  \end{claim}

   \claimproof
   It is enough to notice that $\leftcomp{\xi}(c) = \leftcomp{\rho_l}\Pi\alpha h(c) = \leftcomp{\rho}_l(b)$ and $\rightcomp{\xi}(c) = \rightcomp{\rho}_l(b)$, so the conclusion indeed follows form the definition of $X$.

   \smallskip
   
   The next claim follows from the constraints (Cl) and (Cr).

   \begin{claim}
      Suppose that $\{c,c'\}$ is an edge in the Gaifman graph of $\rel C$.
      If $\leftcomp{\xi}(c) = \leftcomp{\xi}(c')$, then $u\langle c\rangle =u\langle c'\rangle $.
      If $\rightcomp{\xi}(c) = \rightcomp{\xi}(c')$, then $v\langle c\rangle =v\langle c'\rangle $.
   \end{claim}
   
   \claimproof
   Since $\{c,c'\}$ is an edge in the Gaifman graph of $\rel C$, there exists $R \in \sigma$ and a tuple of the form $(\dots, c, \dots, c', \dots)$ in $R^{\rel C}$. The $(\alpha h)$-image of this tuple, which is of the form
   $$
   (\dots, (b, (u_i,v_i)_{i=1}^L), \dots, (b', (u'_i,v'_i)_{i=1}^L), \dots)
   $$
   is in $R^{\rel X}$ as $\alpha h$ is a homomorphism from $\rel C$ to $\rel X$. 
   
   If $\leftcomp{\xi}(c) = \leftcomp{\xi}(c')$, then $\leftcomp{\rho}_l(b) = \leftcomp{\rho}_l(b')$ (see the proof of the previous claim), and thus $u_l=u'_l$ by (Cl). But then $u\langle c\rangle =u\langle c'\rangle $ by the definition of $u\langle c\rangle $ and $u'\langle c\rangle $. 

   The second part is proved analogously using (Cr) instead of (Cl) and the proof of the claim is concluded.

   \smallskip

   The proof of soundness is now concluded as follows. 
   Since $\xi$ is grid-like, we know that both $(\leftcomp{\xi})^{-1}(\{i\})$ and $(\rightcomp{\xi})^{-1}(\{i\})$ are connected subsets of the Gaifman graph of $\rel C$ for each $i \in [k]$. It follows from the last claim that $u\langle -\rangle $ is constant on $(\leftcomp{\xi})^{-1}(\{i\})$ and $v\langle -\rangle $ is constant on $(\rightcomp{\xi})^{-1}(\{i\})$. In other words, there exist $\uuu{1}$, $\uuu{2}$, \dots, $\uuu{k} \in G$ and $\vvv{1}$,
   $\vvv{2}$, \dots, $\vvv{k} \in G$ such that
   $$
   u \langle c \rangle = \uuu{\leftcomp{\xi}(c)}, \quad v \langle c\rangle = \vvv{\rightcomp{\xi}(c)} .
   $$
   Since $\xi$ is onto $[k] \times [k]$, the preceding claim then gives us $\uuu{i}=\vvv{i}$ for each $i \in [k]$ and $\{\uuu{i},\vvv{j}\} \in E^{\rel G}$ for each $i,j \in [k]$ with $i \neq j$, therefore $\vvv{1}, \dots, \vvv{n}$ is a $k$-clique in $\rel G$. 

   \medskip

   Finally, notice that $|X|$ is at most $|B| \cdot |G|^{2L}$, which is polynomial in $|G|$, and it then easily follows  that the reduction is an fpt-reduction.
\end{proof}

Our first corollary says that Theorem~\ref{thm:left_hard} indeed covers Grohe's hardness result, Theorem~\ref{thm:hardGrohe}.

\begin{corollary}
Let $\mathcal{C}$ be a recursively enumerable class of structures of bounded arity that does not have bounded tree width modulo homomorphic equivalence. Then $\Gamma = \{(\rel A,\rel A) \mid \rel A \in \mathcal{C}\}$ satisfies the assumptions of Theorem~\ref{thm:left_hard} with $L=1$. 
\end{corollary}

\begin{proof}
Let $k \in \mathbb{N}$. By Theorem~\ref{thm:excludedgrid}, there exists $(\rel A, \rel A) \in \Gamma$ such that the Gaifman graph of the core of $\rel A$ has the $(k \times k)$-grid as a minor. Let $\rel A'$ be the core of $\rel A$, $\alpha$ a homomorphism from $\rel A$ to $\rel A'$, $\beta$ a homomorphism from $\rel A'$ to $\rel A$, and $\mu$ a minor map from the $(k \times k)$-grid to $\rel A'$. Let $\rel A''$ be the connected component of $\rel A'$ that contains the image of $\mu$ and assume, without loss of generality, that $\mu$ is onto $\rel A''$. Let $\nu$ be the induced grid-like mapping from $\rel A''$ onto $[k] \times [k]$, that is, $\nu(a)$ is the unique pair $(i,j)$ such that $a \in \mu(i,j)$. We extend $\nu$ arbitrarily to $A'$ and define $\rho_1 = \nu\alpha$.

Let $g$ be a homomorphism from $\rel A$ to $\rel A$. Since $\rel A'$ is a core, the homomorphism $\alpha g \beta$ 
has a right inverse $\gamma: \rel A' \to \rel A'$, i.e., $\alpha g \beta \gamma$ is the identity on $A'$. We define $\rel C = \rel A''$ and $h = \beta \gamma \iota$, where $\iota: \rel A'' \to \rel A'$ is the inclusion map. Now $\rho_1 g h = \nu \alpha g \beta \gamma \iota = \nu \iota$, which is  grid-like by the choice of $\nu$, and condition (*) is verified. 
\end{proof}

The next corollary shows that Theorem~\ref{thm:left_hard} goes beyond the left-hand side restricted CSPs. The assumptions could be made slightly weaker, but they are, in any case, rather restrictive. 

\begin{corollary}
Let $\Gamma$ be a template and $L \in \mathbb{N}$ such that 
\begin{itemize}
    \item $\{\rel A: (\rel A,\rel B) \in \Gamma, \  \rel A \mbox{ is connected}\}$ does not have bounded tree with modulo homomorphic equivalence and 
    \item for every $(\rel A,\rel B)$ there exist injective homomorphisms $g_1, g_2, \dots, g_L: \rel A \to \rel B$ such that, for every homomorphism $g: \rel A \to \rel B$, we have $gh=g_l$ for some homomorphism $h: \rel A \to \rel A$ and some $l \in [L]$. 
\end{itemize}
Then $\Gamma$ satisfies the assumptions of Theorem~\ref{thm:left_hard}. 
\end{corollary}

\begin{proof}
For each $k$ we use Theorem~\ref{thm:excludedgrid} to find $(\rel A, \rel B) \in \Gamma$ such that $\rel A$ is connected and has a $(k\times k)$-grid as a minor. This gives us a grid-like mapping $\nu$ from $\rel A$ onto the $(k \times k)$-grid. For each $l \in [L]$ we take any $\rho_l$ such that $\rho_lg_l = \nu$, which is possible as $g_l$ is injective. For each homomorphism $g: \rel A \to \rel B$, we have $gh=g_l$ for some $h$ and $l$. Since $\rho_lgh = \rho_lg_l = \nu$, the condition (*) is satisfied with $\rel C = \rel A$.   
\end{proof}

The assumptions of the last corollary in particular require a ``small number'' of homomorphisms from $\rel A$ to $\rel B$. Our final observation is that Theorem~\ref{thm:left_hard} can be sometimes applied when the ``number'' of homomorphisms is not bounded by a constant.

\begin{example} 
In this example it will be convenient to shift the vertex set of a  $(k\times k)$-grid to $\{0,1, \dots, k-1\} \times \{0,1, \dots, k-1\}$ (and change the definition of grid-like mapping accordingly).

Let $\sigma$ be a signature containing a binary symbol $R$ and let $\Gamma = \{(\rel A_k, \rel B_k)|k\in\mathbb{N}\}$, where $A_k$ with $R^{\rel A_k}$ is a $(k \times k)$-grid (in particular, $A_k = \{0,1, \dots, k-1\}^2$),
$B_k$ with $R^{\rel B_k}$ is an $(f(k) \times f(k))$-grid for some $f(k) \geq k$, and every homomorphism from $\rel A_k$ to $\rel B_k$ is injective. 

We show that the assumptions of Theorem~\ref{thm:left_hard} are satisfied with $L=1$. For each $k$ we take $(\rel A,\rel B) = (\rel A_k,\rel B_k)$ and set $\rho_1(i,j) = (i \mod k, j \mod k)$ where $(i,j) \in B_k$.
Because of $R$, every (injective) homomorphism $g: \rel A \to \rel B$ is of the form $g(i,j) = (l+si,l'+s'j)$ for some $l,l' \in \{0,1, \dots \}$ and $s,s' \in \{-1,1\}$, or of the form $g(i,j) = (l+sj,l'+s'i)$. In all the cases, $\rho_1g$ is a grid-like mapping from $\rel A$ onto $\{0,1, \dots, k-1\} \times \{0,1, \dots, k-1\}$, so (*) is satisfied with $\rel C = \rel A$ and the identical $h$.
\end{example}

\subsection{Approximating clique}

In this final subsection we briefly discuss the $f$-$\gapc$ problems. Recall that $f: \mathbb{N} \to \mathbb{N}$ is a  function such that $f(n) \leq n$ for each $n \in \mathbb N$ and that $f$-$\gapc$ is equivalent to $\PHOM(\{(\rel K_{f(k)},\rel K_k) \mid k \in \mathbb{N}\})$. In this subsection, we implicitly assume that all functions and sets are computable.

For the identity function $f$, $f$-$\gapc$ is \emph{p}-$\clique$, it is therefore a $W[1]$-hard problem. A well known open question is how small can $f$ be made.

\begin{question} \label{q:gap_clique}
    For what functions $f$ is $f$-$\gapc$ W[1]-hard? Is it W[1]-hard for any unbounded function $f$?
\end{question}

Note that $f$-$\gapc$ can be used, instead of $\clique$, as a starting point in the proof of Theorem~\ref{thm:left_hard}. It follows from the proof that, whenever $f$-$\gapc$ is W[1]-hard, (*) can be weakened to ``\dots $\rho_lgh$ is a grid-like mapping from $\rel C$ onto $K \times K$ with $K \geq f(k)$'' (where the definition of a grid-like mapping is naturally extended). 

A recent breakthrough toward answering Question~\ref{q:gap_clique} is the following result of Lin~\cite{Lin2021ConstantAK}.

\begin{theorem} \label{thm:approx_clique}
    For any $0<c \leq 1$, the problem  $f$-$\gapc$ is W[1]-hard whenever $f(n) \geq cn$ for all $n \in \mathbb{N}$.
\end{theorem}

The result was further improved in~\cite{CSK22}.

\smallskip

Another natural question is what happens if we consider templates $\{(\rel K_{f(k)},\rel K_k)\}$ where $k$ runs through some infinite set instead of the whole $\mathbb{N}$. Here is a simple observation in this direction.

\begin{proposition} \label{prop:sparse}
Let $f,g: \mathbb{N} \to \mathbb{N}$ be functions such that $f(n),g(n) < n$ for each $n \in \mathbb{N}$, and let $L \subseteq \mathbb{N}$.
Suppose that for every $k \in \mathbb{N}$ there exists $l \in L$ such that $g(l) \geq (l/k + 1) f(k)$. Then $f$-$\gapc$ is fpt-reducible to $\PHOM(\{(\rel K_{g(l)},\rel K_l) \mid l \in L\})$.
\end{proposition}

\begin{proof}
  Given an instance $(\rel K_{f(k)}, \rel K_{k}, \rel G)$ of $f$-$\gapc$ we find $l \in L$ such that $g(l) \geq (l/k+1) f(k)$ and take the smallest integer $m$ such that $mk \geq l$. Note that $m \leq l/k + 1$. We map the given instance to the instance $(\rel K_{g(l)}, \rel K_l, \rel H)$ where
  $\rel H$ is obtained by taking $m$ disjoint copies of $\rel G$ and making all pairs of vertices in different copies adjacent.

  If $\rel G$ contains a $k$-clique, then $\rel H$ contains an $mk$-clique for we can take the same clique in every copy of $\rel G$ in $\rel H$ to get such a clique.
  Since $l \leq mk$, the graph $\rel H$ contains an $l$-clique. This proves the completeness of the reduction.  
  
  Assume now that $\rel H$ contains a $g(l)$-clique. By taking the largest intersection of that clique with a copy of $\rel G$ in $\rel H$, we obtain a clique in $\rel G$ of size at least $g(l)/m$. Since $g(l)/m \geq g(l)/(l/k+1) \geq f(k)$, the graph $\rel G$ contains an $f(k)$-clique, proving soundness of the reduction.
\end{proof}

The proposition allows us to slightly refine Theorem~\ref{thm:approx_clique}.

\begin{corollary}
    For any $0<c \leq 1$ and any infinite $L \subseteq \mathbb{N}$, the problem $\PHOM(\{(\rel K_{g(l)},\rel K_l) \mid l \in L\})$  is W[1]-hard whenever $g(l) \geq cl$ for all $l \in L$.
\end{corollary}

\begin{proof}
    We define $f(n) = (c/2)n$
    and for each $k$ we take any $l \in L$ with $l \geq k$.
    Since $g(l) \geq cl \geq c(l/k+1)k/2 = (l/k+1)f(k)$, Proposition~\ref{prop:sparse} gives us a reduction from $f$-$\gapc$, which is W[1]-hard by Theorem~\ref{thm:approx_clique}.
\end{proof}

\section{Conclusion}

We introduced the framework of left-hand side restricted PCSPs, which simultaneously generalizes left-hand side restricted CSPs and approximation versions of the $k$-clique problem, and we provided some initial results.
The main technical contribution is the sufficient condition for W[1]-hardness in Theorem~\ref{thm:left_hard} which, in particular, covers left-hand side restricted bounded arity CSPs. 
However, it remains to be seen whether this general framework for left-hand side restriction can be as fruitful as it is for the right-hand side restrictions~(see~\cite{BBKO21}). A challenging problem in this direction is to improve the sufficient condition so that it not only covers CSPs but also the constant factor approximation of $k$-clique stated in Theorem~\ref{thm:approx_clique}. Such a result seems to require a significantly different construction.

\bibliography{lhspcsp}

\appendix

\end{document}